\documentclass[journal]{IEEEtran}
\usepackage{cite}
\usepackage{graphicx,subfigure,color}
\usepackage{algorithm}
\usepackage{algorithmic}
\usepackage{amsmath,amssymb}
\usepackage[active]{srcltx}
\DeclareMathOperator*{\argmax}{arg\,max}

\def\expect{{\mathbb{E}}}

\newcommand{\BlackBox}{\rule{1.5ex}{1.5ex}}  
\newenvironment{proof}{\par\noindent{\bf Proof\
}}{\hfill\BlackBox\\[2mm]}

\newtheorem{lemma}{\bf{Lemma}}
\newtheorem{proposition}{Proposition}
\newtheorem{corollary}{Corollary}

\newtheorem{assumption}{Assumption}

\newcommand{\be}{\begin{equation}}
\newcommand{\ee}{\end{equation}}
\newcommand{\bea}{\begin{eqnarray*}}
\newcommand{\eea}{\end{eqnarray*}}

\makeatletter
\def\Lddots{\mathinner{\mkern1mu\raise17\p@\vbox{\kern17\p@\hbox{.}}\mkern2mu
    \raise8\p@\hbox{.}\mkern2mu\raise\p@\hbox{.}\mkern1mu}}
 \makeatother

\outer\def\subsect#1\par{\vskip12pt
plus.07\vsize\penalty-250\vskip0pt plus-.07\vsize
\bigskip\vskip\parskip\message{#1}
\vbox{\smash{\lower9pt\hbox{\kern-8pt\epsfbox{shadedbox.eps}}}}\vskip-\baselineskip
\leftline{\large\bf#1}\nobreak\medskip}

\def\BibTeX{{\rmfamily B\kern-.05em{\scshape i\kern-.025em b}\kern-.08em \TeX}}

\newcommand{\A}{\mathcal{A}}

\newcommand{\N}{\mathcal{N}}

\newcommand{\bB}{\mathbf{B}}

\newcommand{\bX}{\mathbf{X}}

\newcommand{\bs}{\mathbf{s}}

\newcommand{\bx}{\mathbf{x}}

\def\expect{{\mathbb{E}}}

\makeatletter
\newcommand*{\coloneq}{\mathrel{\rlap{%
                     \raisebox{0.3ex}{$\m@th\cdot$}}%
                     \raisebox{-0.3ex}{$\m@th\cdot$}}%
                     =}
\makeatother

\title{The Cognitive Compressive Sensing Problem}

\author{Saeed Bagheri and Anna Scaglione\\
Electrical and Computer Engineering Department\\
University of California, Davis\\
Email: \{sabagheri, ascaglione\}@ucdavis.edu}
  
\begin{document}

\maketitle
\begin{abstract}{

In the Cognitive Compressive Sensing (CCS) problem, a Cognitive Receiver (CR)  seeks to optimize the reward obtained by sensing an underlying $N$ dimensional random vector, by collecting at most $K$ arbitrary projections of it. The $N$ components of the latent vector represent sub-channels states, that change dynamically from ``busy'' to ``idle'' and vice versa, as a Markov chain that is biased towards producing sparse vectors. 
To identify the optimal strategy we formulate the Multi-Armed Bandit Compressive Sensing (MAB-CS) problem, generalizing the popular Cognitive Spectrum Sensing model, in which the  CR can sense $K$ out of the $N$ sub-channels, as well as the typical static setting of Compressive Sensing, in which the CR observes $K$ linear combinations of the $N$ dimensional sparse vector. The CR opportunistic choice of the sensing matrix should balance the  desire of revealing the state of as many dimensions of the latent vector as possible, while not exceeding the limits beyond which the vector support is no longer uniquely identifiable. 
}
\end{abstract}


\section{Introduction}

The multi-armed bandit (MAB) problem models the situation of an agent whose intent is maximizing his long term reward by strategically choosing an {\it arm}, that corresponds to a possible reward. 
A popular application of this framework has been that of Cognitive Spectrum Sensing (CSS) (see e.g. \cite{survey,Zhao,Twkin,Zhao2,Ahmad,Liu,Unnikrishnan,Ahmad2,nonBay-1,nonBay-2}). In CSS, a Cognitive Receiver (CR) can only sense $K$ out of $N$ bands, switching dynamically the bands that are filtered and down $K$ Analog to Digital Converters (ADCs), with the goal of exploiting as frequently as possible the bands left idle by the primary users. 

In our scenario, we borrow the Bayesian formulation of the problem, in which the transition of the channels from ``busy" to ``idle"  is a Markov chain (typically, $N$ independent two state Markov chains) with a known transition probability matrix \cite{Zhao,Twkin,Zhao2,Ahmad,Liu,Unnikrishnan,Ahmad2}. In this case, the model falls in the class of restless multi-armed bandit (RMAB) problems \cite{Zhao2,Ahmad,RMAB-1}.  
There is also a non-Bayesian formulation of the problem \cite{nonBay-1,nonBay-2}, which is not considered in this work.  

 If idle channels are prevalent, it is also natural to attack the problem using ideas from Compressive Sensing (CS) and Finite Rate of Innovation (FRI) sampling \cite{CSS-1,CSS-2,CSS-3,CSS-4,FRI-1,FRI-2,FRI-3,FRI-4,CS-1,CS-2,CS-3} as possible receiver architectures. In this case, the receiver applies the static policy of observing $K$ linear combinations of the set of channels, and relies on the fact that sparse vectors can be recovered uniquely, even for an underdetermined system, i.e. $K<N$. For a given sensing budget $K$, the static architecture is however capable of recovering supports that are half as large as $K$ or larger. Because sparsity is not guaranteed, this receiver architecture is too inflexible to work in practice. 

{\bf Contributions}: The idea in this paper is to overcome the limitations of compressive sensing in these (and potentially in other) sensing applications by merging the MAB online learning formulation with that of CS to form a more general model, making CS an adaptive and cognitive algorithm. In our paper, we envision that the CR can opportunistically activate different linear combinations of the entries of the latent vector. Each compressive row of the sensing matrix is equivalent to an {\it arm} in the MAB problem, and the objective is to select a sensing matrix with $K$ rows and with columns that are a subset of $\{1,\ldots,N\}$ that provides the optimum long term reward. We consider a relatively general formulation of the problem with simplifying assumptions on the sensing model that help make the problem tractable and shed insights on the optimum policy structure. Specifically, we assume noiseless sensing and a sparse vector recovery algorithm that recovers the support perfectly or incurs in an erasure when the identifiability conditions are violated.  Under these assumptions, we derive optimum myopic strategies for the CR and give sufficient conditions on the state space $\{\Omega[t]\}_{t=1}^T$ under which the optimum solution is the greedy policy.

The paper is organized as follows. First, in Section \ref{sec.problem-formulation}, we formulate the Cognitive Compressive Sensing problem in a general form. Then we set the stochastic optimization framework in Section \ref{sec:StocasticOPT}. 
This section is followed by our analysis of the myopic policy and by numerical results that validate our theoretical results.

{\it Notation}: The set of real, complex and integer numbers
numbers by $\mathbb{R}$, $\mathbb{C}$ and $\mathbb{Z}$, respectively. We denote sets by calligraphic symbols,
where the intersection and the union of two sets $\mathcal{A}$ and $\mathcal{B}$
are written as $\mathcal{A}\cap\mathcal{B}$ and $\mathcal{A}\cup\mathcal{B}$, respectively. The operator $|\mathcal{A}|$ on a discrete set takes the cardinality (measure) of the set and $\mathcal{A}^c$ denotes the complement of $\mathcal{A}$, where the universal set should be evident from the context.  We denote vectors and matrices by boldface lower-case and boldface upper-case symbols. The transpose, conjugate, Hermitian (conjugate) transpose, inverse and pseudo inverse of a matrix $\bX$ are denoted by $\bX^T$, $\bX^*$, $\bX^H$, $\bX^{-1}$ and $\bX^{\dagger}$, respectively. The conventional $\ell_2$-norm is written as $\Vert\bx\Vert_2$, the $\ell_1$-norm is denoted by $\Vert\bx\Vert_1$ and $\Vert\bx\Vert_0$ is the number of non-zero entries of the vector $\bx$. The operator $\expect\{\cdot\}$ denotes the expectation operator.

\section{Problem Formulation}\label{sec.problem-formulation}
In this section, we cast the observation model in a general form, idealizing some aspects associated with the physical sensing and considering the case of a linear observation model analogous to that often used in the CS literature \cite{CS_theory}. We also describe the general structure of the stochastic optimization problem that separates the Cognitive CS policy from the prior art.

\subsection{Observation and Decision Model}
The observation model is described pictorially in Fig. \ref{fig.1}. In our model, the CR can only accrue $K$ observations per unit of time that are linear combinations of the elements of an $N$-dimensional random vector $\boldsymbol{\alpha}[t]$, whose non-zero entries are referred to as being ``active'' or ``busy'', while the zero entries are referred to as ``idle'' or ``empty''.
 \begin{figure}[!htb]
  \centering
 \centerline{\includegraphics[width=0.5\textwidth]{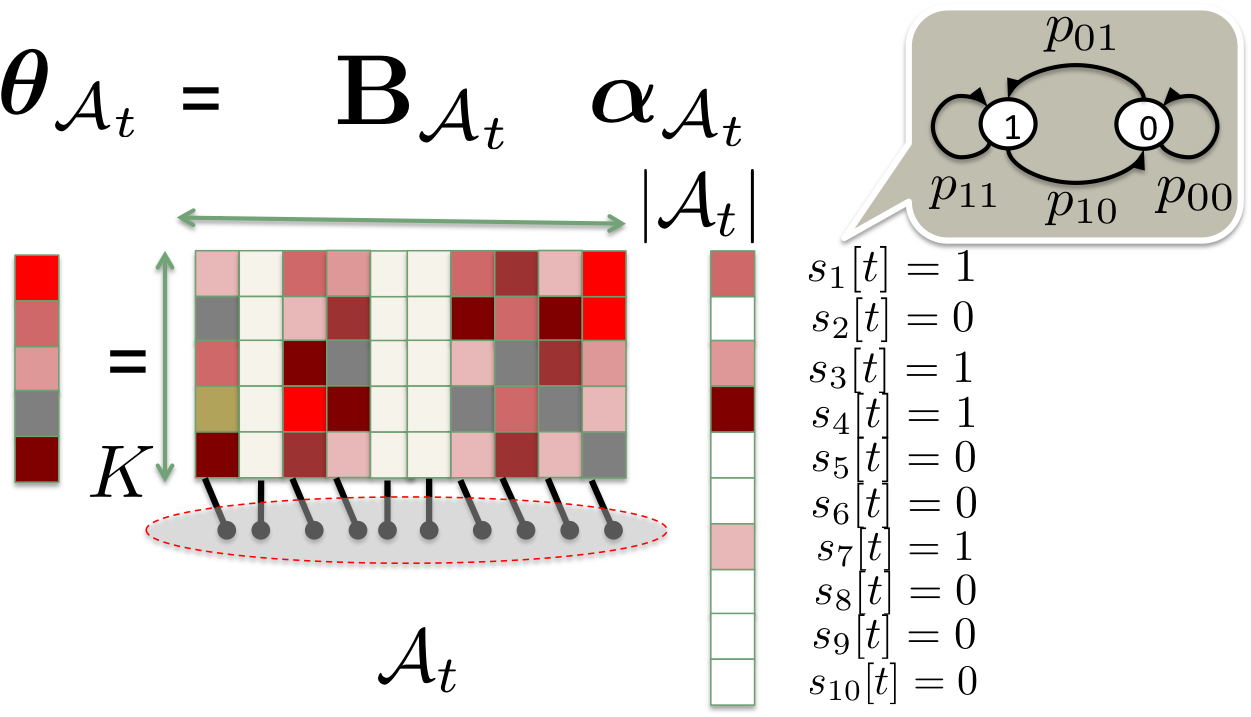}}
\caption{CCS Noiseless Observation Model.}
\label{fig.1}
\end{figure}

The statistics of $\boldsymbol{\alpha}[t]$ are governed by a state vector $\bs[t] = [s_1[t], \ldots, s_N[t]] \in \{0,1\}^N$ modeled as a vector Markov chain and the Bayesian formulation focuses on the statistics of $\bs[t]$. Denoting by $\omega_i[t]=\text{Pr}(s_i[t]=0)$, a reasonable model is that ${\alpha}_i[t]$ is that of a mixture distribution
 $f_{\alpha_i[t]}(\alpha)=\omega_i[t]f_0(\alpha)+(1-\omega_i[t])f_1(\alpha)$. In order for the CS algorithm performance to resemble what we later postulate in Assumption $2$, we assume that if $s_i[t]=0$ the entries ${\alpha}_i[t]$ are zero ($f_0(x)=\delta(x)$) and if $s_i[t]=1$, with probability one $|{\alpha}_i[t]|>\alpha^*$ ($f_1(x)=0$ for $|x|\leq \alpha^*$) where $\alpha^*>0$ and is sufficiently high for their detection. 
In particular,  $s_i[t]$ denotes the state of the $i$-th entry ${\alpha}_i[t]$ in time slot $t$, which is $0$ if the entry is ``idle'' with probability one, i.e. $\text{Pr}({\alpha}_i[t]=0|s_i[t]=0)=1$, and $1$ if the entry is ``busy'', which means that there is $a>0$ such that $\text{Pr}(|{\alpha}_i[t]|>a|s_i[t]=1)=1$. The interest in compressive sensing is associated with cases where the parameters of the model bias the vector towards being sparse.

In general, the Markov chain $\bs[t]\in\{0,1\}^N$ has $2^N$ possible states. The evolution of the states can be represented by the following algebraic equations:
\begin{align}\label{eq.MC}
\boldsymbol{\sigma}[t]&=\boldsymbol{P}\boldsymbol{\sigma}[t-1]+\boldsymbol{\nu}[t]\\
\bs[t]&=\boldsymbol{\Phi}\boldsymbol{\sigma}[t]
\end{align}
where $\boldsymbol{\Phi}$ is an $N\times 2^N$ matrix that contains all possible binary sequences, $\boldsymbol{P}$ is the transition probability matrix, $\boldsymbol{\sigma}[t]$ is the state, that can be one of the coordinate vectors in $\mathbb{R}^{2^N}$, $\boldsymbol{\pi}[t]=\expect\{\boldsymbol{\sigma}[t]\}$ is the state probability at time $t$
and the martingale $\boldsymbol{\nu}[t]=\boldsymbol{\sigma}[t]-\boldsymbol{P}\boldsymbol{\sigma}[t-1]$ allows to write the evolution of the Markov chain as linear dynamical state equations. 
  In our analysis, we introduce the further simplification that each sub-channel evolves as an i.i.d., two-state discrete time Markov chain with transition probabilities $p_{ij}$, $i,j \in \{0,1\}$\footnote{As far as the analysis in this paper are concerned, assuming equal transition probabilities for all $N$ Markov chains is not necessary and the main results hold true for the case with different transition probabilities. We have assumed equal probabilities only to simplify the presentation of the paper.}, which means that we reduce \eqref{eq.MC} into $N$ independent set of equations analogous to \eqref{eq.MC}, except that each of the state vectors $\boldsymbol{\sigma}_i[t]$ is $2\times 1$.   In the case of i.i.d. Markov chains, the vector $\bs[t]$ will tend to be sparse if $p_{10} > p_{01}$. 
  
 The CR objective is to collect a reward which is a function of the number of entries identified correctly as idle and the dilemma is how to ``sense'' the entries by choosing if the data budget should be spent observing individual entries directly or linear combinations of entries. Mathematically, we model the action at time $t$ as the choice of a subset $\A_t$ of columns of a $K\times N$ sensing matrix $\bB$. The action equivalently selects the $K\times |\A_t|$ matrix $\bB_{\A_t}$, where $K\leq |\A_t|\leq N$ denotes the cardinality of the subset of columns of $\bB$ with $\mathcal{A}_t \subseteq \mathcal{N}$.
 In the following, we will use interchangeably $\mathcal{A}_t $ and $\mathcal{N}$ as sets of columns for the sensing matrix as well as sets of indices $\mathcal{A}_t \subset \mathcal{N}\triangleq \{1,\ldots,N\}$. We also use the index $i$ to denote the $i$-th sensing column, and use $+$ and $-$ in lieu of $\cup$ and $\cap$ for sets.   
 We indicate the set of the parts of $\mathcal{N}$ as $2^{\mathcal{N}}$, so that $\mathcal{A}_t\in 2^{\mathcal{N}}$.  
 We introduce the following assumption:
\begin{assumption}({\bf CCS}){\it~ 
We assume that for any $\A_t$, the columns of $\bB_{\A_t}$ are drawn of a set of vectors so that any $K$ columns out of the $|\A_t|>K$ columns of $\bB_{\A_t}$ form a linearly independent set. For $|\A_t|=K$, this means simply that $\bB_{\A_t}$  is full rank. For $K<|\A_t|\leq N$, this condition ensures that  vectors $\boldsymbol{\alpha}_{\A_t}$ with sparsity $<K/2$ can be identified uniquely \cite{CS_theory}.  }
\end{assumption}  

Meeting Assumption 1 is as difficult as finding a $K\times N$ matrix  $\bB$ that has the same property
and selecting the subset $\A_t$ of its columns to form $\bB_{\A_t}$, as shown in Fig. \ref{fig.1}, where the MAB arms are the columns selected from a $K\times N$  sensing matrix. More precisely:
\begin{proposition}{\it 
 The action space for the CR is a matroid \cite{edmond} $\mathcal{M}(\mathcal{N},2^{\mathcal{N}})$, with ground set the columns of the sensing matrix associated to $\mathcal{N}$ and $2^{\mathcal{N}}$ as the collection of {\it independent} sets. }
 \end{proposition}

In our notation, the $|\A_t|\times 1$ vector $\boldsymbol{\alpha}_{\mathcal{A}_t}$ is the potentially sparse vector whose entries are $\alpha_i[t]$ for indices in the set $\{i \in \A_t: s_i[t]=1\}$.  We define $\bs_{\A_t}$ as the support of the vector $\boldsymbol{\alpha}_{\A_t}$, where $\bs_{\A_t}$ includes the entries in $\bs[t]$ corresponding to the indices in the set $\A_t$ and assume that the reward 
 that the CR seeks is a function of $\bs[t]$ only. 
 The CR task is to recover the vector $\boldsymbol{\alpha}_{\A_t}$ and its support $\bs_{\A_t}$ (observable system state) based on the observation vector $\boldsymbol{\theta}_{\mathcal{A}_t}$. With the definitions given above, the noiseless observation model is expressed as (see Fig. \ref{fig.1})
\begin{equation}
\label{obs_model}
\boldsymbol{\theta}_{\mathcal{A}_t} = \bB_{\A_t}\boldsymbol{\alpha}_{\A_t}\,.
\end{equation} 

The $\ell_1$ norm of $\|\bs_{\A_t}\|_1$ is the number of non-zero entries of $\boldsymbol{\alpha}_{\A_t}$ ($\|\boldsymbol{\alpha}_{\A_t}\|_0$).  If $|\A_t|=K$, then any full rank sensing matrix would lead to the exact recovery of $\boldsymbol{\alpha}_{\A_t}$. For $|\A_t|>K$,
 a well established fact in compressive sensing \cite{CS_theory} is that a necessary and sufficient condition to recover uniquely $\boldsymbol{\alpha}_{\A_t}$ is that  $\|\bs_{\A_t}\|_1 < K/2$ and that any $K$ columns of the sub-matrix $\bB_{\A_t}$ are linearly independent. In the absence of noise, the choice of the specific coefficients of $\bB_{\A_t}$ can be purely based on widening the observability for $\boldsymbol{\alpha}_{\A_t}$, but otherwise it can be completely arbitrary, thanks to Assumption $1$.
  
In this work, we do not delve into the details of the sparse recovery algorithm \cite{CS_theory}, but consider an idealized version of the data processing that conforms to the following characteristics:
\begin{assumption}{\it
As long as the number of non-zero entries in $\bs_{\A_t}$ is smaller than  $ K/2$ ($\|\bs_{\A_t}\|_1 < K/2$), our sparse recovery algorithm is able to uniquely recover the support vector $\bs_{\A_t}$ from the observation $\boldsymbol{\theta}_{\A_t}$. If $\|\bs_{\A_t}\|_1 \geq K/2$, the CR experiences an erasure, collects no reward and no information from the action $\A_t$, as if it were the empty set. }
\end{assumption}  

Assumption $2$ is able to capture the well known phase transition in the behavior of sparse recovery algorithms, sharpening the effects of the transition.

\section{Stochastic optimization framework}\label{sec:StocasticOPT}

The system state in slot $t$ is not observable due to the constraint that $K\leq N$. Hence the stochastic optimization is an instance of the general model of Partially Observable Markov Decision Processes (POMDPs) \cite{Zhao}. For a POMDP, a sufficient statistic for the optimal decisions is the conditional probability that each state entry is in state $0$ (idle) given all past decisions and observations \cite{Zhao}. We denote the vector whose $i$-th entry is the conditional probability that $s_i[t] = 0$ by $\boldsymbol{\Omega}[t]\triangleq [\omega_1[t],\ldots,\omega_N[t]]$ and refer to it as the {\it belief vector}. The belief vector $\boldsymbol{\Omega}[t+1]$ can be updated recursively, given the action selected  $\A_t$ and the observation $\boldsymbol{\theta}_{\A_t}$ in slot $t$.

Compared to the prior art on spectrum sensing, the action $\A_t$ is no longer limited to have a cardinality $|\A_t| = K$, but is any $\A_t \in 2^{\mathcal N}$, with $K \leq |\A_t| \leq N$. This means that we 
have significantly expanded the action space and the observation space. Another complication arises when the solution is not unique: the belief vector update $\boldsymbol{\Omega}[t+1]\triangleq \mathcal{T}(\boldsymbol{\Omega}[t]|\A_t,\boldsymbol{\theta}_{\mathcal{A}_t}[t])$ is not straightforward, if it is based on the actual information provided by the  model \eqref{obs_model}. This is where Assumption $2$ helps: 
given the sensing action $\A_t$ and the observations $\boldsymbol{\theta}_{\mathcal{A}_t}$ in slot $t$, in our work, we obtain the belief vector for slot $t+1$ based on the value of $\Vert\bs_{\A_t}\Vert_1$ and whether there is an erasure or not. Thanks to this drastic simplification, the belief update depends only on $\bs_{\A_t}$ rather than $\boldsymbol{\theta}_{\mathcal{A}_t}[t]$, i.e.: 
\begin{equation}\label{eq.beliefupd}
\boldsymbol{\Omega}[t+1]= 
\mathcal{T}(\boldsymbol{\Omega}[t]|\A_t,\bs_{\A_t})\;,
\end{equation}
where \eqref{eq.beliefupd} is instrumental to make our problem tractable.  
More specifically, to tackle the behavior of the CR output in general, we define the following integer threshold function:
\begin{equation}
\Gamma_{\A_t}=\left \{ \begin{array}{ll} 
|\A_t|, & |\A_t|\leq K\\
 \lceil K/2\rceil - 1, & K<|\A_t|\leq N 
\end{array}\right..
\end{equation}
Under Assumption $2$, the belief update 
$\boldsymbol{\Omega}[t+1] = \mathcal{T}(\boldsymbol{\Omega}[t] | \A_t, \bs_{\A_t})$ in terms of the {\it recovered support vector} $\bs_{\A_t}$ and $\Gamma_{\A_t}$ leads to the following expression for each $\omega_i[t+1]$ \begin{align}
\omega_i[t+1] = \left \{ \begin{array}{ll} 
p_{10}, & i \in \A_t,\: \|\bs_{\A_t}\|_1 \leq \Gamma_{\A_t},\: s_i(t) = 1\\
p_{00}, & i \in \A_t,\: \|\bs_{\A_t}\|_1 \leq \Gamma_{\A_t},\: s_i(t) = 0\\
\tau(\omega_i[t]), & i \in \A_t,\: \|\bs_{\A_t}\|_1 > \Gamma_{\A_t}\\
\tau(\omega_i[t]). & i \notin \A_t
\end{array}\right. ,
\end{align}
where $\tau(\omega) \triangleq \omega p_{00} + (1 - \omega)p_{10}$. Note that when $|\A_t|=K$, the condition $\|\bs_{\A_t}\|_1 \leq \Gamma_{\A_t}$ is always met. 

Next we define two set functions that subsume all the important metrics guiding the CR to operations. One is the probability of the erasure event postulated in Assumption $2$, and thus is $f: {\mathcal N}\mapsto [0,1]$, the second is the reward function and is $f: {\mathcal N}\mapsto \mathbb{R}^+$. One useful definition which we use frequently in the following is as follows: 
given a set function, $f: {\mathcal N}\mapsto \mathbb{R}$ we define its marginal increment as
\begin{equation}\label{eq.marginal}
\frac{\partial f(\A)}{\partial a}\triangleq f(\A+a)-f(\A).
\end{equation}

\subsection{Erasure Event and its Probability}
In our problem, the CR decisions are entirely guided by its beliefs on the random variable $\|\bs_{\A_t}\|_1$. One of the aspects that the CR has to mind in its choice is the possibility of an erasure. 
The erasure event for action $\A_t$ is described as
\begin{equation}
{\mathcal E}_{\A_t}=(\|\bs_{\A_t}\|_1 > \Gamma_{\A_t}).
\end{equation}
The PMF of $\Vert\bs_{\A_t}\Vert_1$ is denoted by $P_{\|\bs_{\A_t}\|_1 }( k )$ and its CDF and complementary CDF are represented by $F_{\|\bs_{\A_t}\|_1 }(x)$ and $F^c_{\|\bs_{\A_t}\|_1 }(x)$, respectively. Using the law of total probability for conditional probabilities, we can obtain the PMF of the random variable $\|\bs_{\A_t+a}\|_1$ recursively:
\be 
\label{eq.PMF_recursion}
P_{\|\bs_{\A_t +a}\|_1 }( k ) = \left \{ \begin{array}{ll}
\omega_a[t]P_{\|\bs_{\A_t}\|_1 }( k )\,, & k=0\\
(1-\omega_a[t])P_{\|\bs_{\A_t}\|_1 }( k -1) \\ + \omega_a[t]P_{\|\bs_{\A_t}\|_1 }( k )\,. & 1\leq k \leq \Gamma_{\A_t}
\end{array}\right. .
\ee
Based on Assumption $2$, the probability of erasure given the action $\A_t$ can be written as follows
\begin{align}
\rho_{\A_t}&\triangleq\text{Pr}({\mathcal E}_{\A_t}|\A_t)=F^{c}_{\|\bs_{\A_t}\|_1 }(\Gamma_{\A_t}) \nonumber \\
&=\left \{ \begin{array}{ll} 
0, & |\A_t|\leq K\\
\text{Pr}(\|\bs_{\A_t}\|_1 > \Gamma_{\A_t}), & K<|\A_t|\leq N
\end{array}\right. .
\end{align}
The expression for $\rho_{\A_t}$ when $|\A_t|>K$ can be equivalently expressed in terms of the PMF of the random variable $\|\bs_{\A_t}\|_1$ as follows
\be
\label{eq.rho_pmf}
\rho_{\A_t} = 1 - \sum_{k = 0}^{\Gamma_{\A_t}} P_{\|\bs_{\A_t}\|_1 }( k ).
\ee
Note that the function $\rho_{\A_t}$ is a set function $f: \{0,1\}^N\mapsto [0,1]$.
Two simple observations follow from the description of the probability of erasure. First, considering the fact that $\text{Pr}(\|\bs_{\A_t}\|_1 > \Gamma_{\A_t})=
\text{Pr}(\|\bs_{\A_t-i}\|_1+s_i[t] > \Gamma_{\A_t})$, the function $\rho_{\A_t}$ obeys the following formula for all $i\in \A_t$:
\begin{equation}\label{eq.rhoCDF}
\rho_{\A_t}=\left \{ \begin{array}{ll} 
0, & |\A_t|\leq K\\
\omega_i[t]F^{c}_{\|\bs_{\A_t-i}\|_1 }(\Gamma_{\A_t}) \\+(1-\omega_i[t])F^{c}_{\|\bs_{\A_t-i}\|_1 }(\Gamma_{\A_t}-1) & |\A_t|=K+1\\
(1-\omega_i[t])F^{c}_{\|\bs_{\A_t-i}\|_1 }(\Gamma_{\A_t}-1)\\+\omega_i[t]\rho_{\A_t-i}~, & K+1<|\A_t|\leq N
\end{array}\right. 
\end{equation}
Using the expressions in \eqref{eq.PMF_recursion} and \eqref{eq.rho_pmf}, we can easily derive the marginal increment of $\rho_{\A_t}$ as
\begin{align}
\label{eq.rho_marginal}
\frac{\partial \rho_{\A_t} } {\partial a}&\triangleq \rho_{\A_t+a} - \rho_{\A_t} \nonumber \\
&= \left \{ \begin{array}{ll} 
 \rho_{\A_t+a} & |\A_t|=K\\
 (1-\omega_a[t])P_{\|\bs_{\A_t}\|_1}(\Gamma_{\A_t}) & K<|\A_t|\leq N-1
 \end{array}\right. .
\end{align}
Also, one can prove the following Lemmas which present the critical properties of the erasure probability.
\begin{lemma}\label{lemma.erasure_risk}{\it  The set function $\rho_{\A}$ is monotone non-decreasing with respect to $\A$, i.e., if ${\mathcal A} \subset {\mathcal B}$ then $\rho_{\A}\leq \rho_{\mathcal B}$. Furthermore, given a set $\A_t$, if $\omega_i[t]>\omega_j[t]$ then the risk of erasure is such that
\be\label{eq.rho-order}
\rho_{\A_t+i}\geq \rho_{\A_t+j}.
\ee
}
\end{lemma}
\begin{proof}
The first statement is a mere result of the axioms of probability. In addition, we can directly show it by using \eqref{eq.rho_marginal} successively and add the marginal increments to $\rho_{\A}$ due to the elements in $\mathcal{B}\setminus \A$. The second statement is trivially satisfied for $|\A_t|=K$ and for $|\A_t|\geq K+1$, using \eqref{eq.rho_marginal} we can write 
\begin{align}
\rho_{\A_t+i}&=\rho_{\A_t}+\frac{\partial \rho_{\A_t} } {\partial i}\nonumber\\
		    &=\rho_{\A_t}+ (1-\omega_i[t])P_{\|\bs_{\A_t}\|_1}(\Gamma_{\A_t}).    
\end{align}
As a result, we have
\be
\rho_{\A_t+i}-\rho_{\A_t+j}=(\omega_i[t]-\omega_j[t])	P_{\|\bs_{\A_t}\|_1}(\Gamma_{\A_t})\geq 0,
\ee
where the last inequality follows from $\omega_i[t]>\omega_j[t]$.
\end{proof}
\begin{lemma}\label{lemma.rho-supmod}
{\it If $\boldsymbol{\Omega}[t]$ is such that $P_{\|\bs_{\A_t}\|_1}(\Gamma_{\A_t})$ is a monotonic non-decreasing set function with respect to $\A_t$, then the set function $\rho_{\A_t}$ is super-modular. That is, if  $\A_t\subset {\mathcal B}_t$, $a\notin \A_t$ and $a\notin {\mathcal B}_t$:
\begin{equation}
\frac{\partial \rho_{{\mathcal B}_t}}{\partial a}\geq \frac{\partial \rho_{\A_t}}{\partial a}.
\end{equation}}
\end{lemma}
\begin{proof}
The case where $|\A_t|=K$ is simple because, for any non-trivial state $\boldsymbol{\Omega}[t]$, there is a discontinuity and the erasure probability jumps from zero to a positive value just by adding any extra element. For  $|\A_t|>K$,
if $P_{\|\bs_{\A_t}\|_1}(\Gamma_{\A_t})$  is monotonic then for all $\A_t\subset {\mathcal B}_t$ we have $P_{\|\bs_{{\mathcal B}_t}\|_1}(\Gamma_{\A_t})\geq P_{\|\bs_{\A_t}\|_1}(\Gamma_{\A_t})$.
Using \eqref{eq.rho_marginal}, we obtain
\begin{align*}
\frac{\partial \rho_{{\mathcal B}_t}}{\partial a}-\frac{\partial \rho_{\A_t}}{\partial a}
= (1-\omega_a[t])\left(P_{\|\bs_{{\mathcal B}_t}\|_1}(\Gamma_{\A_t})-P_{\|\bs_{\A_t}\|_1}(\Gamma_{\A_t})\right),
\end{align*}
which is positive under the assumption that $P_{\|\bs_{\A_t}\|_1}(\Gamma_{\A_t})$  is monotonic non-decreasing.
\end{proof}

Lemma \ref{lemma.erasure_risk} captures mathematically the intuitive fact that the CR incurs in increasing risk as it selects bigger sets. The chances of reward are tempered by this fact, as discussed in the next Section. 
Lemma \ref{lemma.rho-supmod} is much stronger, showing that the risk grows faster than linearly and that is true if bigger sets raise the probability $P_{\|\bs_{\A_t}\|_1}(\Gamma_{\A_t})$ of the support value that is at the boundary between success and erasure. However, noting that $0\leq \rho_{\A}\leq 1$, as the erasure probability grows, its increments have to taper off to zero eventually. The condition of  Lemma \ref{lemma.rho-supmod} amounts to assuming that 
the inflection does not happen for any of the subsets in $2^{\mathcal N}$.

\subsection{Reward Function}
The Cognitive CS optimization over a finite horizon $T$ is a decision process driven by the objective of maximizing the average reward over a horizon of $T$ slots, by choosing strategically a sensing policy that governs the selection of the observation in each slot.  A policy is a function that sifts through all possible options represented by $T$-uples of subsets in $(2^{\mathcal N})^T$ each having cardinality greater or equal to $K$. More rigorously, a sensing policy $\pi$ is a sequence of functions $\pi \triangleq [\pi_1, \ldots, \pi_T ]$, where $\pi_t$ is the decision rule at time $t$ that maps a belief vector $\boldsymbol{\Omega}[t]$ onto a sensing action $\A_t \subseteq \mathcal{N}$, i.e. $\mathcal{A}_t = \pi_t(\boldsymbol{\Omega}[t])$. The optimal policy accrues the maximum total expected reward over a finite horizon, i.e.
\begin{equation}
\label{optimal_policy}
\pi^* = \argmax_{\pi}~~\expect_{\pi} \left[\sum_{t=1}^T R_{\pi_t}\big(\boldsymbol{\Omega}[t]\big){\Bigg\vert} \boldsymbol{\Omega}[1] \right],
\end{equation}
where $R_{\pi_t}(\boldsymbol{\Omega}[t])$ is the reward obtained at time $t$ corresponding to the belief vector $\boldsymbol{\Omega}[t]$ using the policy $\pi_t$. For any given $\mathcal{A}_t =\pi_t(\boldsymbol{\Omega}[t])$, 
the reward is indicated as $R_{\A_t}(\boldsymbol{\Omega}[t])$. 

For a given sensing policy $\pi$, the belief vectors $\{\boldsymbol{\Omega}[t]\}_{t=1}^{T}$ form a Markov process with an uncountable state space. The expectation in \eqref{optimal_policy} is with respect to this Markov process which determines the reward process. The vector $\boldsymbol{\Omega}[1]$ is the initial belief vector and if no information about the initial system state is available, each entry of $\boldsymbol{\Omega}[1]$ can be set to the stationary distribution $\omega_o$ of the underlying Markov chain:
\be 
\omega_o = \dfrac{p_{10}}{p_{01}+p_{10}}.
\ee 

In the following, we introduce two possible formulations for the reward functions. 
\subsubsection{The CCS idle channels collector}
For spectrum sensing applications, the idle channels bring reward, since they can be used to communicate among secondary users, without interfering with a primary user. In the most basic instance of this, with
equal bandwidth for subbands, when a channel is detected to be idle, the CR can collect one unit of reward. If none of the channels sensed is in the idle state or if there is an erasure, the CR collects no reward, and waits until the next time slot to make another choice. Mathematically, the reward of taking action $\A_t$ is expressed as
\begin{align} 
\label{eq:rewardCCS-collector}
R_{\A_t}[t] &= \sum_{i\in \A_t} (1-s_i[t]).
\end{align}

When $|\A_t| = K$, Assumption $2$ states that $\bB_{\A_t}$ is full rank and, therefore, the CR can uniquely recover $\boldsymbol{\alpha}_{\A_t}$ and its support vector $\bs_{\A_t}$. This is equivalent to the MAB problem with $K$-arms posed in \cite{Ahmad2}. But the structure of the reward is different for $|\A_t|>K$, due to the fact that the reward is collected only if the support is smaller than $K/2$. 
The following Lemma describes an expression for the expected immediate reward.
\begin{lemma}\label{lemma:expected_reward1}{\it Under Assumptions $1$ and $2$, the expected reward of taking action $\A_t$ for the CR is}
\begin{align} \label{eq:expected_reward}
\expect\big[R_{\A_t}[t]\big]&= \sum_{i\in \A_t}\omega_i[t]F_{\|\bs_{\A_t-i}\|_1}(\Gamma_{\A_t}).
\end{align}
\end{lemma}
\begin{proof}
The expected immediate reward of the CSS idle channels collector can be written as
\begin{align} 
\label{eq:expected_rewardproof}
\expect\big[R_{\A_t}[t]\big] &= \sum_{i\in \A_t} \text{Pr}((s_i[t]=0)\cap {\mathcal E}_{\A_t}^c) \nonumber \\
					      &=\sum_{i\in \A_t} \text{Pr}((s_i[t]=0)\cap (\|\bs_{\A_t-i}\|_1+s_i[t]\leq \Gamma_{\A_t})) \nonumber \\
					      &=\sum_{i\in \A_t} \text{Pr}(s_i[t]=0)\text{Pr}(\|\bs_{\A_t-i}\|_1 \leq \Gamma_{\A_t}),
\end{align}
where ${\mathcal E}_{\A_t}^c$ denotes the complement of the erasure event given the action $\A_t$. 
\end{proof}

The following Lemmas provide a few key characteristics of the expected immediate reward that are used in our analysis in Section \ref{sec:optimal_policy}. 
\begin{lemma} \label{lemma.reward-submod}
{\it  For $|\A_t|>K$, the marginal rewards associated with \eqref{eq:rewardCCS-collector} is:
\be\label{eq.marginalreward}
\frac{\partial \expect\big[R_{\A_t}[t]\big]}{\partial a}
=\omega_a[t](1-\rho_{\A_t})-\sum_{i\in \A_t}\omega_i[t]\frac{\partial F^c_{\|\bs_{\A_t-i}\|_1}(\Gamma_{\A_t})}{\partial a},
\ee
and for $|\A_t|=K$, the marginal reward is:
\begin{align}
\label{eq.marginalrewardAt=K}
\frac{\partial \expect\big[R_{\A_t}[t]\big]}{\partial a}&=\omega_a[t](1-F_{\|\bs_{\A_t}\|_1}^{c}(\lceil K/2\rceil-1))\nonumber \\
&-\sum_{i\in \A_t} \omega_i[t] F_{\|\bs_{\A_t+a-i}\|_1}^{c}(\lceil K/2\rceil-1).
\end{align}
Moreover, if $\boldsymbol{\Omega}[t]$ is such that $P_{\|\bs_{\A_t}\|_1}(\Gamma_{\A_t})$ is a monotonic non-decreasing set function with respect to $\A_t$, the expected reward is a sub-modular set function, i.e. for all $\A_t\subset {\mathcal B}_t$
$$\frac{\partial \expect\big[R_{\A_t}[t]\big]}{\partial a}\geq \frac{\partial \expect\big[R_{{\mathcal B}_t}[t]\big]}{\partial a}.$$}
\end{lemma}
\begin{proof}
For any $a\notin \A_t$ with $|\A_t|>K$, the marginal reward is:
\begin{align*} 
\frac{\partial \expect\big[R_{\A_t}[t]\big]}{\partial a} &= \sum_{i\in \A_t} \text{Pr}((s_i[t]=0)\cap {\mathcal E}_{\A_t+a}^c)\nonumber \\
&+\text{Pr}(s_a[t]=0)\text{Pr}(\|\bs_{\A_t}\|_1 \leq \Gamma_{\A_t})\\
						&- \sum_{i\in \A_t} \text{Pr}((s_i[t]=0)\cap {\mathcal E}_{\A_t}^c)\\
						&=\sum_{i\in \A_t}\omega_i[t](F_{\|\bs_{\A_t-i}\|_1}^{c}(\Gamma_{\A_t})-F_{\|\bs_{\A_t+a-i}\|_1}^{c}(\Gamma_{\A_t}))\nonumber \\&+\omega_a[t](1-\rho_{\A_t}),
\end{align*}
which leads to the expression in \eqref{eq.marginalreward}. The proof of \eqref{eq.marginalrewardAt=K} is along the same lines, except that $\rho_{\A_t}\neq F_{\|\bs_{\A_t}\|_1}^{c}(\lceil K/2\rceil-1)$ due to the discontinuity in the erasure probability. For $|\A_t|>K$, returns are diminishing as shown next:
\begin{align} 
&\frac{\partial \expect\big[R_{\A_t}[t]\big]}{\partial a}-\frac{\partial \expect\big[R_{{\mathcal B}_t}[t]\big]}{\partial a}
					= \omega_a[t](\rho_{{\mathcal B}_t}-\rho_{\A_t})\nonumber \\ &+\sum_{i\in {\mathcal A}_t}\omega_i[t]
					\left(\frac{\partial F^c_{\|\bs_{\mathcal{B}_t-i}\|_1}(\Gamma_{\A_t})}{\partial a}-\frac{\partial F^c_{\|\bs_{\A_t-i}\|_1}(\Gamma_{\A_t})}{\partial a}\right),
\label{eq.marginal_reward_diff}
\end{align}
where in the last equation, the first term in the Right-Hand Side (RHS) is positive according to Lemma \ref{lemma.erasure_risk}. Using \eqref{eq.PMF_recursion}, we can express $\frac{\partial F^c_{\|\bs_{\A_t-i}\|_1}(\Gamma_{\A_t})}{\partial a}$ as
\begin{align} 
\label{eq.marginal_cdf}
\frac{\partial F^c_{\|\bs_{\A_t-i}\|_1}(\Gamma_{\A_t})}{\partial a} = (1-\omega_a[t])P_{\|\bs_{\A_t-i}\|_1}(\Gamma_{\A_t}).
\end{align}
Then the second term in RHS of \eqref{eq.marginal_reward_diff} can be written as 
\be 
\sum_{i\in \A_t}\omega_i[t](1-\omega_a[t])(P_{\|\bs_{\mathcal{B}_t-i}\|_1}(\Gamma_{\A_t})-P_{\|\bs_{\A_t-i}\|_1}(\Gamma_{\A_t})),
\ee 
where all the terms in the summation are positive if $P_{\|\bs_{\A_t}\|_1}(\Gamma_{\A_t})$ is a monotonic non-decreasing set function with respect to $\A_t$ for $|\A_t|> K$.
\end{proof}
\begin{lemma}
\label{lemma.max_marginal} {\it For any given set $\A_t$, the marginal reward $\frac{\partial \expect\big[R_{\A_t}[t]\big]}{\partial a}$ is maximized by adding the element 
\be 
a = \argmax_{i \in \mathcal{N}\setminus \A_t} ~  \omega_i[t].
\ee}
\end{lemma}
\begin{proof}
We consider the case of $|\A_t|>K$, but the proof is easily generalized to the case $\|\A_t\|=K$.
From \eqref{eq.marginalreward}, the marginal reward of adding $a$ to set $\A_t$ is equal to
$$
\frac{\partial \expect\big[R_{\A_t}[t]\big]}{\partial a}=\omega_a[t](1-\rho_{\A_t})-\sum_{i\in \A_t}\omega_i[t]\frac{\partial F^c_{\|\bs_{\A_t-i}\|_1}(\Gamma_{\A_t})}{\partial a}.
$$
Replacing \eqref{eq.marginal_cdf} in \eqref{eq.marginalreward} results in the following expression for the marginal reward
\begin{align}
\frac{\partial \expect\big[R_{\A_t}[t]\big]}{\partial a} &=\omega_a[t](1-\rho_{\A_t})-\sum_{i\in \A_t}\omega_i[t](1-\omega_a[t])P_{\|\bs_{\A_t-i}\|_1}(\Gamma_{\A_t})  \nonumber \\
&= \omega_a[t]\left(1-\rho_{\A_t} + \sum_{i\in \A_t}\omega_i[t]P_{\|\bs_{\A_t-i}\|_1}(\Gamma_{\A_t})\right) \nonumber \\ &- \sum_{i\in \A_t}\omega_i[t]P_{\|\bs_{\A_t-i}\|_1}(\Gamma_{\A_t}).
\end{align}
Since $1-\rho_{\A_t} + \sum_{i\in \A_t}\omega_i[t]P_{\|\bs_{\A_t-i}\|_1}(\Gamma_{\A_t}) \geq 0$, the marginal reward is maximized whenever $\omega_a[t]$ is the maximum possible value in the set $\N\setminus \A_t$ which completes the proof.
\end{proof}

\subsubsection{The CCS busy channels collector}
The MAB formulation is useful also in sensing applications aimed at detecting and tracking signal activities in a set of sub-channels. In this case, the CR earns a unit of reward for each non-zero entry in $\boldsymbol{\alpha}[t]$ that is detected correctly:
\be\label{eq:rewardCCS-busy}
R_{\A_t}[t]=\sum_{i\in \A_t} s_i[t].
\ee

Lemma \ref{lemma:css_busy} introduces an expression for the the expected immediate reward.
\begin{lemma}
\label{lemma:css_busy} {\it The expected immediate reward for \eqref{eq:rewardCCS-busy} given the action $\A_t$ is
\be 
\label{eq:expected_reward2}
\expect\big[R_{\A_t}[t]\big] = |\A_t|(1-\rho_{\A_t}) - \sum_{i\in \A_t}\omega_i[t]F_{\|\bs_{\A_t-i}\|_1}(\Gamma_{\A_t}).
\ee }
\end{lemma}
\begin{proof}
In this case:
\begin{align*}
\expect\big[R_{\A_t}[t]\big] 
 &= \sum_{i\in \A_t} \text{Pr}((s_i[t]=1)\cap {\mathcal E}_{\A_t}^c)\\
					      &=\sum_{i\in \A_t} \text{Pr}((s_i[t]=1)\cap (\|\bs_{\A_t-i}\|_1+s_i[t] \leq \Gamma_{\A_t}))\\
					      &=\sum_{i\in \A_t} \text{Pr}(s_i[t]=1)\;\text{Pr}(\|\bs_{\A_t-i}\|_1 \leq \Gamma_{\A_t}-1)\\
					      &=\sum_{i\in \A_t} (1-\omega_i[t])(1-F^{c}_{\|\bs_{\A_t-i}\|_1 }(\Gamma_{\A_t}-1))\\
					      &\stackrel{a}{=}\sum_{i\in \A_t} (1-\omega_i[t])-\sum_{i\in \A_t} [\rho_{\A_t} - \omega_i[t] F^c_{\|\bs_{\A_t-i}\|_1}(\Gamma_{\A_t})]\\
					      & = |\A_t|(1-\rho_{\A_t}) - \sum_{i\in \A_t}\omega_i[t]F_{\|\bs_{\A_t-i}\|_1}(\Gamma_{\A_t}).
\end{align*}
where $a$ follows from \eqref{eq.rhoCDF}. 
\end{proof}

The interesting observation about the structure of the reward in this case is that for no region in the state space this function is sub-modular or super-modular, and this makes the problem NP-hard. It means that to find the myopic policy, the CR must compute the expected reward for all the possible actions $\A_t$ with $K\leq |\A_t|\leq N$. This is completely in contrast to the CCS idle channels collector case as we see in the analysis of the myopic policy in the next Section.

\section{Study of the Optimal CCS Policy}
\label{sec:optimal_policy}
The maximum expected total reward obtainable starting from slot $t$ given the current belief vector $\boldsymbol{\Omega}[t]$ is the so called value function denoted by $V_t(\boldsymbol{\Omega}[t])$. It includes two terms: the expected immediate reward $\expect[R_{\A_t}[t]]$ and the maximum expected future reward $V_{t+1}(\mathcal{T}(\boldsymbol{\Omega}[t] | \A_t, \bs_{\A_t}))$, given that the user takes action $\A_t$ and the reward attainable through the observation $\boldsymbol{\theta}_{\A_t}$ in slot $t$ is only a function of the observable state $\bs_{\A_t}$, based on Assumption $2$ and \eqref{eq.beliefupd}.  
Averaging over all possible states $\bs_{\A_t}$ and maximizing over all actions $\A_t$, we obtain the following recursion, whose solution corresponds to a dynamic program 
\begin{align}
\label{value_function}
V_T(\boldsymbol{\Omega}[T]) &= \max_{\A_t} ~~~ \expect\big[R_{\A_t}[T]\big] \nonumber \\
V_t(\boldsymbol{\Omega}[t]) &= \max_{\A_t} ~~~ \Big [\expect\big[R_{\A_t}[t]\big] + \sum_{\bs \in \{0,1\}^{|\A_t|}} \text{Pr}(\bs_{\A_t}=\bs |\A_t) \nonumber\\
&~~~~~~~~~~~~~\cdot V_{t+1}\big(\mathcal{T}(\boldsymbol{\Omega}[t] | \A_t, \bs_{\A_t}=\bs)\big)\Big]\,,
\end{align}
where the summation is over all possibilities for the observable state $\bs_{\A_t}$. The optimal policy $\pi^*$ and its performance $V_1(\boldsymbol{\Omega}[1])$ are computationally prohibitive to derive brute force,  since the belief vector $\boldsymbol{\Omega}[t] \in [0,1]^N$ lies in an uncountable set. 

It is a standard step to study first the {\it myopic policy}  $\hat{\pi}$, that is a stationary policy that maximizes the expected immediate reward $\expect\big[R_{\A_t}[t]\big]$ under the current belief vector $\boldsymbol{\Omega}[t]$, disregarding the effect of the current action on the future reward and is expressed as
\be 
\label{eq:myopic_p}
\A^*_t = \argmax_{\A} ~~ \expect\big[R_{\A}[t]\big]\,.
\ee 
Solving \eqref{eq:myopic_p} is also computationally intensive, given that $\expect\big[R_{\A}[t]\big]$ is a set function and $K\leq |\A|\leq N$. In the following, we restrict our focus on the study of the myopic policy, the solution of \eqref{eq:myopic_p}. 

\subsection{Analysis of the myopic CCS policy for the empty channels collector}\label{sec:myopic_policy}
We consider a CR that is an {\it empty channels collector}. 
We denote by $(n_1,\ldots,n_N)$ the permutation of the indices that orders the belief vector as $\omega_{n_1}[t]\geq \omega_{n_2}[t]\geq\ldots\geq\omega_{n_N}[t]$. 
To find the myopic policy, we need to solve the following optimization problem to find the action with $|\A_t| \geq K$:
\be  
\label{eq:opt_A}
\A^*_t  = \argmax_{\A \in 2^{\mathcal N}, |\A|\geq K} ~~ \sum_{i\in \A}\omega_i[t]F_{\|\bs_{\A_t-i}\|_1}(\Gamma_{\A_t}).
\ee 
We first establish an order for the class of sets with fixed cardinality $|\A_t| = M$, i.e. the set:
 \be
 {\mathcal P}^{(M)} \triangleq \{\A: |\A|=M,~\A \in 2^{\mathcal N}\}.
 \ee
For $|\A_t| = K$, considering $\omega_i[t]$ as the weights of each element, it is well known that the maximum weight over a matroid in general is obtained by a greedy algorithm \cite{edmond}:
\begin{align}
\label{eq:R_tilde}
{\mathcal R}^{(K)} (\boldsymbol{\Omega}[t])&\triangleq \max_{\A:~|\A|=K} ~ \expect\big[R_{\A}[t]\big] \nonumber \\ &= \max_{\A} ~ \sum_{ \stackrel{i \in \A}{|\A|=K}} \omega_{i}[t] = \sum_{i=1}^{K}\omega_{n_i}[t]\,,
\end{align}
which corresponds to the set $$\A^{(K)}_t \triangleq \{n_1,n_2,\ldots,n_K\}.$$ Next we extend this property to $ {\mathcal P}^{(M)}$ with $K<M\leq N$ and prove that $${\mathcal R}^{(M)} (\boldsymbol{\Omega}[t]) \triangleq \max_{\A:~|\A|=M} ~ \expect\big[R_{\A}[t]\big]=\expect\big[R_{\A^{(M)}}[t]\big].$$
\begin{lemma}
\label{lemma_1}{\it 
Any set $\A \in  {\mathcal P}^{(M)}$ will have expected reward $\expect\big[R_{\A}[t]\big]$ no larger than the set 
\be\label{eq.AM}
\A^{(M)} = \{n_1,n_2,\ldots,n_{M}\},
\ee
which includes the $M$ components corresponding to the entries with the largest belief values from the vector $\boldsymbol{\Omega}[t]$. }
\end{lemma}  
\begin{proof}
The Lemma can be proven by induction. As we said, it is certainly true for $|\A| = K$ \cite{edmond}. Let us assume that it is true for ${\mathcal P}^{(M)}$, $M\geq K$ and that ${\mathcal R}^{(M)} (\boldsymbol{\Omega}[t])= \expect\big[R_{\A^{(M)}}[t]\big]\geq  \expect\big[R_{\A}[t]\big]$ for all $\A \in {\mathcal P}^{(M)}$. We can prove that this has to be true for ${\mathcal P}^{(M+1)}$.   All sets $\A'\in {\mathcal P}^{(M+1)}$ can be formed as $\A' \triangleq \A \cup\{i\}$ starting from a set $\A\in {\mathcal P}^{(M)}$ and adding an $i \in \mathcal{N}\setminus\A$. 
Thanks to Lemma \ref{lemma.max_marginal} and because of the hypothesis made by induction, we have:
\begin{align} 
\expect\big[R_{\A'}[t]\big]&=\expect\big[R_{\A}[t]\big]+\frac{\partial \expect\big[R_{\A}[t]\big]}{\partial i} \nonumber \\
					   &\leq {\mathcal R}^{(M)} (\boldsymbol{\Omega}[t])+\frac{\partial \expect\big[R_{\A}[t]\big]}{\partial n_{M+1}},
\end{align}
and the inequality is replaced by equality if and only if $\A' = \{n_1,\ldots,n_{M+1}\}$ which completes the proof.
\end{proof}

A direct and fundamental consequence of Lemma \ref{lemma_1} is presented in the following Lemma.
\begin{corollary}\label{cor.OPTMYOP_general}{\it 
The optimum myopic policy is to select the action:
\be
\A_t^* \triangleq \{n_1,\ldots,n_{M*}\},
\ee
where (c.f. \eqref{eq.AM}):
\be
M^*=\argmax_{K\leq M\leq N} ~~~ \expect\big[R_{\A^{(M)}}[t]\big].
\label{eq:myopic_opt}
\ee}
\end{corollary}

What is nice about \eqref{eq:myopic_opt} is that we conclude that finding the optimum myopic policy at time $t$ is not NP-hard, it just requires sorting the values of the beliefs, computing ${\mathcal R}^{(M)} (\boldsymbol{\Omega}[t])$ for $K\leq M\leq N$ and finally finding their maximum value. The complexity of this procedure is polylog of $N$. 

There are cases, however, when the process of finding the optimum myopic policy can be even faster by avoiding the computation of ${\mathcal R}^{(M)} (\boldsymbol{\Omega}[t])$ for all $K\leq M\leq N$. Algorithm \ref{alg:greedy} introduces a {\it greedy algorithm} to find the myopic policy which reduces the computational complexity for large values of $N$. Algorithm \ref{alg:greedy} starts from $K$ elements in the optimum set and includes entries until the marginal reward is greater than or equal to $0$ and stops when for the first time the marginal reward value becomes negative. 
\begin{algorithm}
\caption{The Greedy Myopic Algorithm}
\begin{algorithmic}[1] \label{alg:greedy}
\REQUIRE The permutation $(n_1,\ldots,n_N)$ according to $\boldsymbol{\Omega}[t]$
\vspace{3pt}
\STATE {\bf Initialize}: $i=K$ and $\A_t^*=\A^{(K)}=\{n_1,\ldots,n_K\}$. 
\vskip5pt
\STATE {\bf While} $\frac{\partial \expect\big[R_{\A_t^*}[t]\big]}{\partial n_{i+1}}\geq 0$ and $i<N$ 
\vspace{5pt}
\STATE update $\A_t^*=\{n_1,\ldots,n_{i+1}\}$
\vspace{3pt}
\STATE set $i = i+1$
\STATE {\bf Do} 
\end{algorithmic}
\end{algorithm}

By introducing more restrictive conditions, the following corollary establishes a sufficient condition for the greedy algorithm introduced in Algorithm \ref{alg:greedy} to be the optimum myopic policy. 
\begin{corollary}\label{cor.greedy}{\it 
If $\boldsymbol{\Omega}[t]$ is such that $P_{\|\bs_{\A^{(M)}}\|_1}(\Gamma_{\A^{(M)}})$ is a monotonic non-decreasing set function with respect to $M$, then the greedy procedure is the optimal myopic policy. }
\end{corollary}
\begin{proof}
We know from Lemma \ref{lemma.reward-submod} that when $P_{\|\bs_{\A^{(i)}}\|_1}(\Gamma_{\A^{(i)}})$ is monotonic non-decreasing as $i$ increases, $\expect\big[R_{\A^{(i)}}[t]\big]$ is a sub-modular set function. Then, we can show that 
\be
\frac{\partial \expect\big[R_{\A^{(i)}}[t]\big]}{\partial n_{i+1}} \stackrel{(a)}{\geq} \frac{\partial \expect\big[R_{\A^{(i)}}[t]\big]}{\partial n_{j+1}} \stackrel{(b)}{\geq} \frac{\partial \expect\big[R_{\A^{(j)}}[t]\big]}{\partial n_{j+1}}, ~~~ \forall j\geq i
\ee 
where $(a)$ is concluded from Lemma \ref{lemma.max_marginal} knowing that $\omega_{n_{i+1}} \geq \omega_{n_{j+1}}$ and $(b)$ follows from the sub-modularity of the expected reward since $\A^{(i)}\subset \A^{(j)}$.
As a result, $\frac{\partial \expect\big[R_{\A^{(i)}}[t]\big]}{\partial n_{i+1}}$ is monotonic non-increasing w.r.t. $i$. Since $\frac{\partial \expect\big[R_{\A^{(i)}}[t]\big]}{\partial n_{i+1}}$ is monotonic non-increasing, at each step, the greedy policy (Algorithm \ref{alg:greedy}) includes the element that makes the set maximize the expected reward over all candidates in ${\mathcal P}^{(i+1)}$ and it stops when increasing  $i$ further, decreases the expected reward compared to ${\mathcal R}^{(t)}(\boldsymbol{\Omega}[t])$, knowing that from that point on all marginal rewards are indeed negative.  
\end{proof}

\section{Numerical Experiments}
In this Section, we evaluate numerically the performance of the myopic policy for CCS architecture and specifically compare it with the myopic policy for the $K$-arm selection problem \cite{Ahmad2} where the CR selects exactly $K$ out of the $N$ sub-channels to sense at each time slot $t$. In \cite{Ahmad2}, the authors have shown that for $K$-arm selection problem, the myopic policy is optimal when $p_{00} \geq p_{10}$.

\begin{figure}[!htb]

\begin{minipage}[b]{1.0\linewidth}
  \centering
  \centerline{\includegraphics[width=0.95\textwidth]{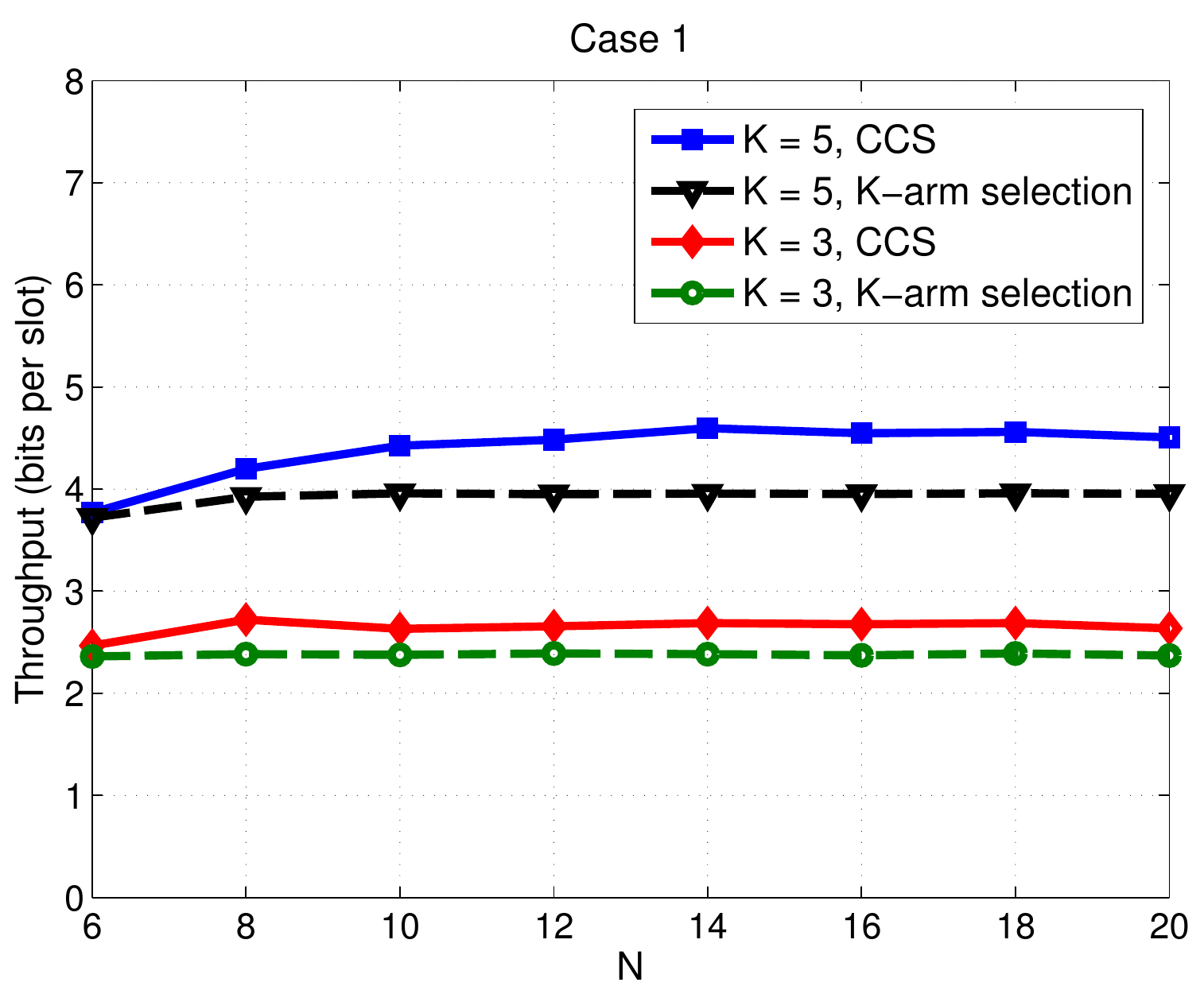}}
  \centerline{(a) Case $1$: $p_{10} = 0.42$ and $p_{00} = 0.82$}\medskip
\end{minipage}
\begin{minipage}[b]{1\linewidth}
  \centering
  \centerline{\includegraphics[width=0.95\textwidth]{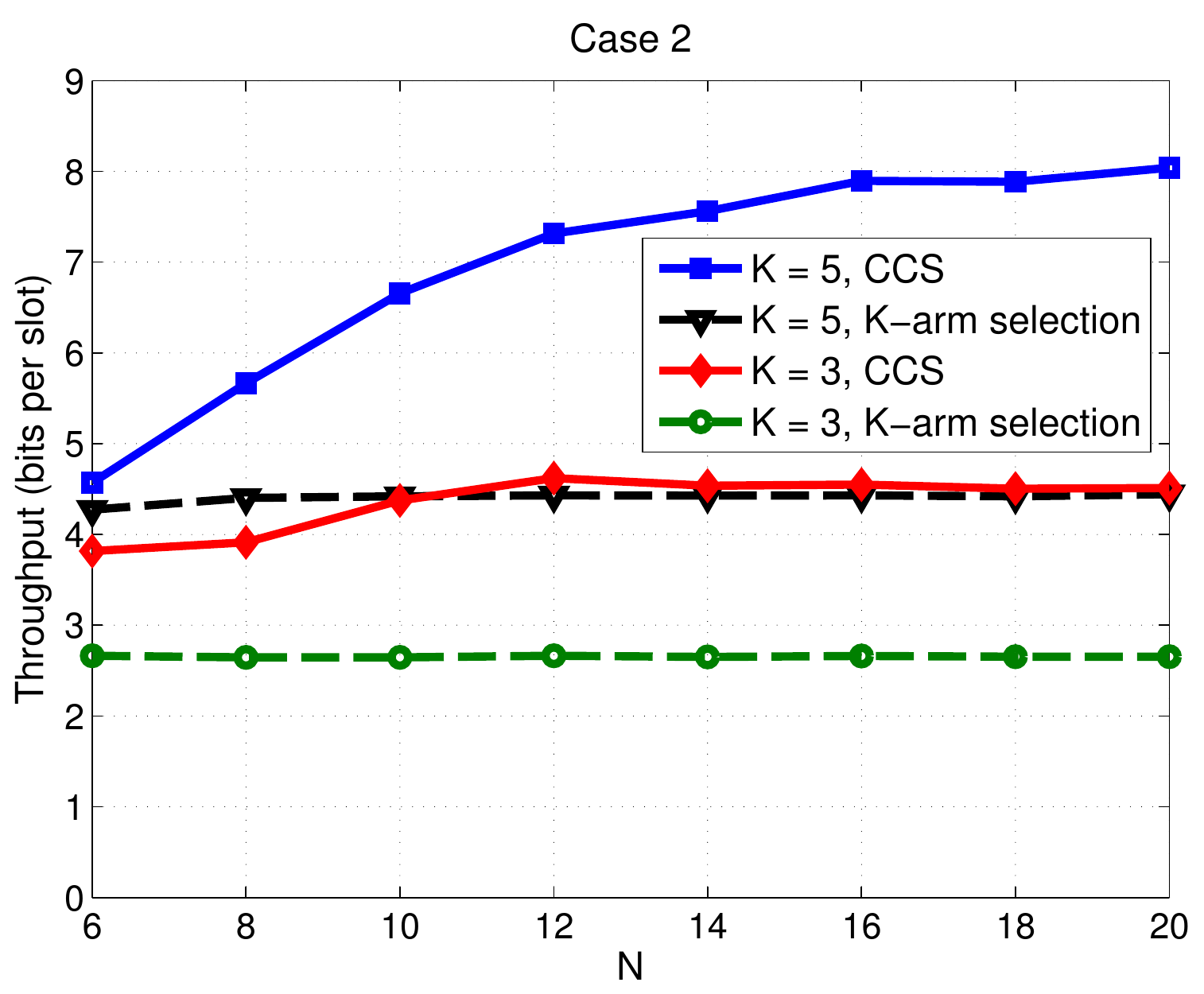}}
  \centerline{(b) Case $2$: $p_{10} = 0.4$ and $p_{00} = 0.9$}\medskip
\end{minipage}
\caption{Performance of the myopic policy in CCS and $K$-arm selection.}
\label{fig_cases}
\end{figure}

In the numerical experiments, we assume $T=30$ and the number of arms is equal to $K=3$ and $5$. We consider $N$ independent sub-channels with the same transition probabilities and bandwidth $B = 1$. The value of $N$ is set to vary from $6$ to $20$ and we compute the normalized expected total reward achieved over $500$ simulation trials. For better comparison and visualization reasons, the results are normalized by $T$ to reflect the throughput per slot. We consider two scenarios for the transition probabilities to capture the sparsity in spectrum occupancy and study the effect of the sparse channel occupancy on the performance of CCS. In Case $1$, we set the transition probabilities as $p_{10} = 0.42$ and $p_{00} = 0.82$, which in the steady state corresponds to spectrum occupancy rate of $30\%$. In Case $2$, we investigate a sparser scenario with transition probabilities $p_{10} = 0.4$ and $p_{00} = 0.9$, which in the steady state corresponds to channel occupancy rate of $20\%$.
 
In Fig. \ref{fig_cases}(a), the performance of myopic policy for CCS and $K$-arm selection are presented for $K=3, 5$ in Case $1$. The myopic policy in CCS outperforms the myopic policy in $K$-arm selection for all values of $N$ and for both $K=3$ and $5$. Fig. \ref{fig_cases}(b) shows the performance comparison for Case $2$. In this case, with sparser channel occupancy, the performance improvement is more significant. We also observe that myopic CCS with $K=3$ outperforms $K$-arm selection with $5$ arms when $N\geq 10$. The experiments showcase the capability of CCS architecture to improve the expected total throughput when the channel occupancy is sparse. Evidently, the myopic policy in CCS problem requires more processing and is more computationally extensive. However, our experiments demonstrate that in sparse enough settings (e.g. Case $2$), it can considerably enhance the expected throughput. 

\section{Conclusion}
In this paper, we combined the perspective of MAB with FRI sampling structure. We specifically formulated the selection of a compressive sensing {\it arm} with $K$ branches as a MAB problem. We assumed that when the number of active subbands in the selected subset to sense is limited by $K/2$, the states of the sensed sub-channels are perfectly identifiable. For the complexity reduced and noiseless CCS problem we considered in this work, the myopic policy was established and investigated numerically. The numerical experiments demonstrate that in finite horizon setting and when the channel occupancy is sparse, exploiting sparsity in CCS problem improves the expected total reward.

\bibliographystyle{IEEEbib}

\end{document}